\documentclass[12pt]{amsart}
\usepackage{txfonts}
\usepackage{amssymb}
\usepackage{eucal}
\usepackage{graphicx}
\usepackage{amsmath}
\usepackage{amscd}
\usepackage[all]{xy}           

\usepackage{amsfonts,latexsym}
\usepackage{xspace}
\usepackage{epsfig}
\usepackage{float}
\usepackage{color}
\usepackage{fancybox}
\usepackage{colordvi}
\usepackage{multicol}
\usepackage{colordvi}
\usepackage{stmaryrd}
\usepackage[colorlinks,final,backref=page,hyperindex,hypertex]{hyperref}

\newtheorem{theorem}{Theorem}[section]
\newtheorem{defn}[theorem]{Definition}
\newtheorem{lemma}[theorem]{Lemma}
\newtheorem{coro}[theorem]{Corollary}
\newtheorem{prop-def}{Proposition-Definition}[section]

\newcommand{\nc}{\newcommand}
\newcommand{\delete}[1]{}

\nc{\mlabel}[1]{\label{#1}}  
\nc{\mcite}[1]{\cite{#1}}  
\nc{\mref}[1]{\ref{#1}}  
\nc{\mbibitem}[1]{\bibitem{#1}} 

\delete{
\nc{\mlabel}[1]{\label{#1}  
{\hfill \hspace{1cm}{\bf{{\ }\hfill(#1)}}}}
\nc{\mcite}[1]{\cite{#1}{{\bf{{\ }(#1)}}}}  
\nc{\mref}[1]{\ref{#1}{{\bf{{\ }(#1)}}}}  
\nc{\mbibitem}[1]{\bibitem[\bf #1]{#1}} 
}

\nc{\bfk}{\mathbf{k}}
\nc{\Der}{\mathrm{Der}}
\nc{\Ker}{\mathrm{Ker}}

\begin{document}

\title{Homogeneous Rota-Baxter operators on $3$-Lie algebra $A_{\omega}$ }

\author{RuiPu  Bai}
\address{College of Mathematics and Information  Science,
Hebei University, Baoding 071002, China} \email{bairuipu@hbu.edu.cn}

\author{Yinghua Zhang}
\address{College of Mathematics and Information  Science,
Hebei University, Baoding 071002, China} \email{zhangyinghua1234@163.com}

\date{\today}

\begin{abstract} In the paper we study homogeneous  Rota-Baxter operators with weight zero on the infinite dimensional simple $3$-Lie algebra $A_{\omega}$ over a field $F$ ( $ch F=0$ ) which is realized by an associative commutative
algebra $A$ and  a derivation $\Delta$ and an involution $\omega$ ( Lemma \mref{lem:rbd3} ). A homogeneous  Rota-Baxter operator on $A_{\omega}$ is a linear map $R$ of $A_{\omega}$ satisfying $R(L_m)=f(m)L_m$ for all generators of $A_{\omega}$, where $f : A_{\omega} \rightarrow F$. We proved that $R$ is a  homogeneous  Rota-Baxter operator on $A_{\omega}$ if and only if $R$ is the one of the five possibilities  $R_{0_1}$, $R_{0_2}$,$R_{0_3}$,$R_{0_4}$ and  $R_{0_5}$, which are described in Theorem  \mref{thm:thm1}, \mref{thm:thm4},
\mref{thm:thm01}, \mref{thm:thm03} and \mref{thm:thm04}.    By the five homogeneous  Rota-Baxter operators $R_{0_i}$, we construct new $3$-Lie algebras $(A, [ , , ]_i)$ for $1\leq i\leq 5$,  such that
$R_{0_i}$ is the homogeneous  Rota-Baxter operator on $3$-Lie algebra $(A, [ , , ]_i)$, respectively.

\end{abstract}

\subjclass[2010]{17B05, 17D99.}

\keywords{ $3$-Lie algebra, homogeneous  Rota-Baxter operator, Rota-Baxter $3$-algebra.}

\maketitle



\allowdisplaybreaks

\section{Introduction}

Rota-Baxter operators were originally defined on associative algebras by G. Baxter to
solve an analytic formula in probability ~\cite{Ba} and  populated by
the work of Cartier and Rota~\mcite{Ca,Ro1,Ro2}. They have been closely related to many fields in mathematics and mathematical physics.
Rota-Baxter algebras have played an important role in the Hopf algebra approach of renormalization of perturbative quantum field theory of Connes and Kreimer~\mcite{CK,EGK,EGM}, as well as in the application of the renormalization method in solving divergent problems in number theory~\mcite{GZ,MP}.

 Rota-Baxter operators on a Lie algebra are an operator form of the classical Yang-Baxter equations and contribute to the study of integrable systems~\cite{Bai,BGN,BGN2}.
 Semenov-Tian-Shansky¡¯s fundamental work \cite{STS} shows
that a Rota-Baxter operator of weight 0 on a Lie algebra is exactly the operator form
of the classical Yang-Baxter equation (CYBE), which was regarded as a ¡°classical limit¡±
of the quantum Yang-Baxter equation \cite{BCM3B}. Whereas the latter is also an important topic
in many fields such as symplectic geometry, integrable systems, quantum groups and
quantum field theory~\cite{Ag2,BBGN,EGK,Guw,Gub,GK1, GK3, HuiB, GZ,Ro1,Ro2}.

 Rota-Baxter $n$-algebras and differential $n$-algebras were first introduced in~\cite{BGL11}, they are the generalization of Rota-Baxter algebras to the multiple algebraic systems.  We know that
  $n$-Lie algebras~\cite{F}  are a type of multiple  algebraic
systems appearing in many fields of mathematics and mathematical
physics~\cite{N, T, BL,HHM, HIM, HCK,G,P}. Especially,  $3$-Lie algebras and metric $3$-Lie algebras are applied to the study of the supersymmetry and
gauge symmetry transformations of the world-volume theory of
multiple coincident M2-branes; the Bagger-Lambert theory has a novel
local gauge symmetry which is based on a metric $3$-Lie algebra; the $n$-Jacobi identity in $n$-Lie algebras can be regarded as
a generalized Plucker relation in the physics literature.
The theory of $n$-Lie algebras has been widely studied~\cite{K,L,BSZ,BSZ1, BBW, BHB, AI}. For the recent years, the most interesting work on the structure of $n$-Lie algebras is
the realization of $n$-Lie algebras ( $n\geq 3 $ ) from well know algebras, for example,  from Lie algebras, associative algebras, commutative associative algebras, cubic matrices, etc. \cite{BBW, BW22, BLZ33, Pozhi, Pozhi2, Pozhi3}.

Authors in paper  ~\cite{BGL11} provided the Rota-Baxter operator on $n$-Lie algebras and  studied  the structure of Rota-Baxter $3$-Lie algebras,  and also gave the method to realize Rota-Baxter $3$-Lie algebras from Rota-Baxter $3$-Lie algebras,  Rota-Baxter Lie algebras, Rota-Baxter pre-Lie algebras and Rota-Baxter commutative associative algebras and derivations.

In this paper we investigate a class of Rota-Baxter operators with weight zero on the simple $3$-Lie algebra $A_{\omega}$, which is constructed from a commutative
associative
algebra $A$ and a derivation $\Delta$ and an involution $\omega$ which satisfies $\Delta\omega+\omega\Delta=0$ \cite{BW22}.
This on one hand further studies the structures of the simple Rota-Baxter $3$-Lie algebra,
and on the other hand provides a rich source of examples for Rota-Baxter $3$-Lie algebras.

The article is organized as follows. Section 2 describes concepts of Rota-Baxter operators
with weights for general $n$-ary algebras and some results which are used in the paper.  In Section 3 is devoted to the
homogeneous  Rota-Baxter operators on $A_{\omega}$  with weight zero. At last of the paper, new $3$-Lie algebras are constructed by the homogeneous  Rota-Baxter operators on $A_{\omega}$.

In this paper, we suppose that $F$ is a field of characteristic zero, and $Z$ is the set of integer numbers.

\section{preliminary}
\mlabel{sec:rbn}

An {\bf $n$-Lie algebra} \cite{F} is a vector space $A$ over a field $F$ endowed with an $n$-ary multi-linear skew-symmetric operation
$[x_1, \cdots, x_n]$ satisfying the $n$-Jacobi identity
\begin{equation}
  [[x_1, \cdots, x_n], y_2, \cdots, y_n]=\sum_{i=1}^n[x_1, \cdots, [ x_i, y_2, \cdots, y_n], \cdots,
  x_n].
\mlabel{eq:nlie}
\end{equation}
In particular, a {\bf 3-Lie algebra} is a vector space $A$ endowed with a ternary multi-linear skew-symmetric operation
satisfying
for all $x_1,x_2,x_3, y_2, y_3\in A$.
\begin{equation}
[[x_1,x_2,x_3],y_2,y_3]=[[x_1,y_2,y_3],x_2,x_3] +[[x_2,y_2,y_3],x_3,x_1]+[[x_3,y_2,y_3],x_1,x_2].
\end{equation}

\begin{defn}
\mlabel{de:nlie}  Let $\lambda\in F$ be fixed.
\begin{enumerate}
\item
An {\bf $n$-(nonassociative) algebra} over a field $F$ is a
pair $(A, \langle , \cdots, \rangle)$ consisting of a vector space
$A$ over $F$ and a multilinear
 multiplication
 $$\langle , \cdots, \rangle: A^{\otimes n} \rightarrow A.$$
\item
A {\bf derivation of weight $\lambda$} on an $n$-algebra $(A, \langle , \cdots, \rangle)$ is a
linear map $d: A\rightarrow A$ such that,
\begin{equation}
d(\langle x_1, \cdots,x_n\rangle) =\sum\limits_{\emptyset \neq I\subseteq [n]}\lambda^{|I|-1} \langle \check{d}(x_1),\cdots, \check{d}(x_i),\cdots, \check{d}(x_n)\rangle,
\mlabel{eq:dern}
\end{equation}
where
$\check{d}(x_i):=\check{d}_I(x_i):=\left \{\begin{array}{ll} d(x_i), & i\in I, \\ x_i, & i\not\in I\end{array} \right . \text{ for all } x_1,\cdots,x_n\in A.$
Then $A$ is called a {\bf differential $n$-algebra of weight $\lambda$.}
In particular, a {\bf differential $3$-algebra of weight $\lambda$} is a $3$-algebra $(A,\langle , , \rangle)$ with a linear map $d: A\to A$ such that
\begin{eqnarray}
d(\langle x_1,x_2,x_3\rangle )&=& \langle d(x_1),x_2,x_3\rangle +\langle x_1,d(x_2),x_3\rangle +\langle x_1,x_2,d(x_3)\rangle \notag \\
&&
+\lambda \langle d(x_1),d(x_2),x_3\rangle
+\lambda \langle d(x_1),x_2,d(x_3)\rangle
+\lambda \langle x_1,d(x_2),d(x_3)\rangle\\
&&+\lambda^2 \langle d(x_1),d(x_2),d(x_3)\rangle.
\notag
\end{eqnarray}
\item
A {\bf Rota-Baxter operator of weight $\lambda$} on $(A,\langle,\cdots,\rangle)$ is a linear map $R: A\rightarrow A$ such that
\begin{equation}
\langle R(x_1), \cdots, R(x_n)\rangle
=R\left( \sum\limits_{\emptyset \neq I\subseteq [n]}\lambda^{|I|-1} \langle \hat{R}(x_1), \cdots, \hat{R}(x_i), \cdots, \hat{R}(x_{n})\rangle\right),
\mlabel{eq:rbn}
\end{equation}
where
$\hat{R}(x_i):=\hat{R}_I(x_i):=\left\{\begin{array}{ll} x_i, & i\in I, \\ R(x_i), & i\not\in I \end{array}\right. \text{ for all } x_1,\cdots, x_n\in A.
$
Then $A$ is called a {\bf Rota-Baxter $n$-algebra of weight $\lambda$}.
In particular, a {\bf Rota-Baxter $3$-algebra} is a $3$-algebra $(A,\langle , , \rangle)$ with a linear map $P: A\to A$ such that
\begin{eqnarray}
\langle R(x_1), R(x_2), R(x_3)\rangle
&=& R\Big(
\langle R(x_1), R(x_2),x_3\rangle +\langle R(x_1),x_2, R(x_3)\rangle +\langle x_1, R(x_2), R(x_3)\rangle \notag \\
&&
+\lambda \langle R(x_1),x_2,x_3\rangle
+\lambda \langle x_1, R(x_2),x_3\rangle
+\lambda \langle x_1,x_2, R(x_3)\rangle
\mlabel{eq:rb3de}\\
&&
+\lambda^2 \langle x_1,x_2,x_3)\rangle\Big).
\notag
\end{eqnarray}
\end{enumerate}
\end{defn}

\begin{lemma}\cite{BGL11}
Let $(A, \langle\ , \cdots, \rangle)$ be an $n$-algebra over $F$.
An invertible linear mapping $P: A\rightarrow A$ is a Rota-Baxter
operator of weight $\lambda$ on $A$ if and only if $P^{-1}$ is a
differential operator of weight $\lambda$  on $A$.
 \mlabel{lem:rbd1}
\end{lemma}

\begin{lemma}
Let $(A, \langle\ , \cdots, \rangle, R)$ be a Rota-Baxter $n$-algebra over $F$ with weight $ 0$.
Then for all $\lambda\in F$, $\lambda\neq 0$,  $(A, \langle\ , \cdots, \rangle, \lambda R)$ is a Rota-Baxter $n$-algebra with weight $0$.
\mlabel{lem:rbd2}
\end{lemma}

\begin{proof}

The result follows from Eq. (\mref{eq:rbn}), directly.

\end{proof}

\begin{lemma}\cite{BW22} Let $A$ be a vector space with a basis $\{ L_n~|~ n\in Z \}$ over  field $F$. Then $A$ is a simple $3$-Lie algebra
 in the multiplication
\begin{equation}
{[}L_{l},L_{m},L_{n}]=\begin{vmatrix}
(-1)^l&(-1)^m&
   (-1)^n \\
1&1&1  \\
l& m&
   n \\
\end{vmatrix}L_{l+m+n-1}, ~~\mbox{for all } ~~l, m, n\in Z.
\mlabel{eq:defthlmn}
\end{equation}
\mlabel{lem:rbd3}
\end{lemma}

In the following, the $3$-Lie algebra $A$ in Lemma \mref{lem:rbd3} is  denoted by ${\bf A_{\omega}}$, and  the determinant

\vspace{2mm}  $\begin{vmatrix}
(-1)^l&(-1)^m&
   (-1)^n \\
1&1&1  \\
l& m&
   n \\
\end{vmatrix}$
is  denoted by $D(l, m, n)$.

\begin{lemma} The determinant
 $D(l, m, n)=0$  if and only if

 \vspace{2mm} $(l-m)(l-n)(m-n)=0$, or $l=2k+1, m=2s+1, n=2t+1$,  or $l=2k, m=2s, n=2t$,  for all $k, s, t\in Z.$
\mlabel{lem:rbd4}
 \end{lemma}

\begin{proof}
The result follows from a direct computation.
\end{proof}

\section{Homogeneous Rota-Baxter operators with  weight $0$ on $3$-Lie algebra $A_{\omega}$}
\mlabel{sec:rbl3rbl}

In this section we discuss   Rota-Baxter operators with weight $0$ on the $3$-Lie algebra $A_{\omega}$.

By Definition~\mref{de:nlie}, if $(A, [, , ], R)$ is a Rota-Baxter $3$-Lie
algebra of weight $\lambda=0$. Then the linear map $R: A\rightarrow
A$ satisfies that for all $x_1, x_2, x_3\in A$,
\begin{eqnarray}
[R(x_1),R(x_2),R(x_3)] &=& R\Big( [R(x_1),R(x_2),x_3] +[
R(x_1),x_2,R(x_3)] +[ x_1,R(x_2),R(x_3)] \Big).
\mlabel{eq:rf30}
\end{eqnarray}

 {\bf A homogeneous Rota-Baxter operator $R$}  on the $3$-Lie algebra $A_{\omega}$ is a Rota-Baxter operator
satisfies that there exists $f: Z \rightarrow F$
satisfying
\begin{equation}
R(L_{m})=f(m)L_{m},~~  \forall m\in Z.
\mlabel{eq:rf31}
\end{equation}

\begin{theorem} Let $R: A_{\omega} \rightarrow A_{\omega}$ be a linear map defined as   Eq.~(\mref{eq:rf31}).  Then $R$ is a homogeneous Rota-Baxter operator of weight
$0$ on $A_{\omega}$ if and only if
$f$ satisfies for all $l, m, n\in Z$,

\begin{equation}
f(l)f(m)f(n)D(l, m, n)=(f(l)f(n)+f(m)f(n)+f(l)f(m))f(l+m+n-1)D(l, m, n).
\mlabel{eq:rf32}
\end{equation}
\mlabel{thm: thm0}
\end{theorem}

\begin{proof} By Eqs.~(\mref{eq:defthlmn}), (\mref{eq:rf30}) and (\mref{eq:rf31}), we have
$$
[R(L_{l}), R(L_{m}), R(L_{n})]=f(l)f(m)f(n)D(l, m, n)L_{l+m+n-1},
$$
and
$$
R\Big( [R(L_{l}), R(L_{m}), L_{n}]+[R(L_{l}), L_{m}, R(L_{n})]+[L_{l}, R(L_{m}), R(L_{n})])=
$$
$$
(f(l)f(m)+f(l)f(n)+
f(m)f(n))f(l+m+n-1)D(l, m, n)L_{l+m+n-1}.
$$

Therefore, $R$ is a homogeneous Rota-Baxter operator on $A_{\omega}$ if and only if Eq. (\mref{eq:rf32}) holds.

\end{proof}

\subsection{Homogeneous Rota-Baxter operators with $f(0)+f(1)\neq 0$}
\mlabel{ss:prelie3lie}

In this section we discuss the homogeneous Rota-Baxter operators with weight $0$ defined by Eq. (\mref{eq:rf31}) with $f(0)+f(1)\neq 0$.

\begin{theorem} Let  $R: A_{\omega}\rightarrow A_{\omega}$ be  a linear map defined as Eq. (\mref{eq:rf31}) with $f(0)+f(1)\neq 0$. Then
$R$ is a homogeneous Rota-Baxter operator on $A_{\omega}$ if and only if

 $$f(m)=0, ~~ \text{for all}~~m\in Z, ~~\text{ and }~~  m\neq 0, 1.$$
 \mlabel{thm:thm1}
\end{theorem}

\begin{proof}  If $f$ satisfies $f(m)=0,$ ~ for all~$m\in Z,$  and $ m\neq 0, 1.$ By a direct computation $R$ is a  homogeneous Rota-Baxter operator.

Conversely,  if $R$ is a  homogeneous Rota-Baxter operator with $f(0)+f(1)\neq 0$. Then  Eq.~(\mref{eq:rf30}) of the case  $l=0$, $n=1$  becomes
$$f(0)f(m)f(1)=\{f(0)f(1)+f(m)f(1)
+f(0)f(m)\}f(m), \forall m\in Z, m\neq 0, 1.$$

Since $f(0)+f(1)\neq 0$, we have
$f(m)^{2}=0, $ for all $ m\in Z$ and $m\neq 0, 1.$  The proof is completed.

\end{proof}

\subsection{Homogeneous Rota-Baxter operators with  $f(0)+f(1)=0 $}
\mlabel{ss:prelie3lie}

\begin{lemma}  Let $R: A_{\omega} \rightarrow A_{\omega}$ be a linear map of $A$ defined as
Eq. (\mref{eq:rf31}) with $f(0)+f(1)=0$. Then $R$ is a homogeneous Rota-Baxter operator if and only if
for all $l, m, n\in Z$,
\begin{equation}
f(2l+1)f(2m+1)f(2n)
=(f(2l+1)f(2m+1)+f(2l+1)f(2n)
\mlabel{eq:rf33}
\end{equation}

\hspace{6.3cm}$+f(2m+1)f(2n))f(2l+2m+2n+1), ~~ ~ m\neq l,$

\begin{equation}
f(2l+1)f(2m)f(2n)=(f(2l+1)f(2m)+f(2l+1)f(2n)
\mlabel{eq:rf34}
\end{equation}

\hspace{6.4cm}$+f(2m)f(2n))f(2l+2m+2n),  ~m\neq n.$
\mlabel{lem:rbd5}
\end{lemma}

\begin{proof} For all $l, m, n\in Z$ and  $l\neq m$ and $m\neq n$, the determinant $D(2l+1, 2m+1, 2n)\neq 0$ and $D(2l+1, 2m, 2n)\neq 0$.  Thanks to Eq. (\mref{eq:defthlmn}) and Eq. (\mref{eq:rf31}),
we obtain Eq. (\mref{eq:rf33}) and Eq. (\mref{eq:rf34}).
\end{proof}

\noindent 3.2.1  {\bf Homogeneous Rota-Baxter operators with  $f(0)=-f(1)\neq 0. $}

\vspace{2mm} Now we discuss the case $f(0)+f(1)=0$, but $f(0)\neq 0$. By Lemma \mref{lem:rbd2}, we can suppose $f(0)=1$, then $f(1)=-1.$

\begin{coro}  Let $R$ be a homogeneous Rota-Baxter operator with $f(0)=-f(1)=1$.
Then we have for all $l, m, n\in Z$,

\vspace{2mm}1) $f(2l+1)f(2m+1)=(f(2l+1)+f(2m+1)+f(2l+1)f(2m+1))f(2l+2m+1), l\neq m,$

\vspace{2mm}2) $f(2l+1)f(2m)=(f(2l+1)+f(2m)+f(2l+1)f(2m))f(2l+2m), m\neq 0,$

\vspace{2mm}3) $f(2l+1)f(2n)=(f(2l+1)+f(2n)-f(2l+1)f(2n))f(2l+2n+1), l\neq 0,$

\vspace{2mm}4) $f(2m)f(2n)=(f(2m)+f(2n)-f(2m)f(2n))f(2m+2n), m\neq n.$
\mlabel{cor:cor1}
\end{coro}

\begin{proof}
The result follows from Lemma \mref{lem:rbd5} and $f(0)=-f(-1)=1$.
\end{proof}

\begin{theorem}  Let $R$ be a homogeneous Rota-Baxter operator  with $f(0)=-f(1)=1$.
Then we have
\begin{equation}
f(1-m)+f(m)=0, ~~ \text{for all} ~ m\in Z.
\mlabel{eq:rf35}
\end{equation}
\mlabel{thm:thm2}
\end{theorem}

\begin{proof} According to Corollary \mref{cor:cor1},  for all $n\in Z$ and $n\neq 0$, we have

 $f(2m+1)(f(2m+2n)-f(2m+2n+1))+f(2n)(f(2m+2n)-f(2m+2n+1))$

 \vspace{2mm}$+f(2m+1)f(2n)(f(2m+2n)+f(2m+2n+1))=0.$

Then in the case $m=-n\neq 0$,  we have $f(2m+1)+f(-2m)=0$, and
$$f(2m+1)+f(1-(2m+1))=0.$$ Similarly, we have  $f(1-2(-m))+ f(2(-m))=0$, for all $m\in Z.$
It follows Eq. (\mref{eq:rf35}).

\end{proof}

\begin{coro} If $R$ is a homogeneous Rota-Baxter operator  on $A_{\omega}$ satisfying that $f(0)=-f(1)=1$, and there exist  $k$, $l$, $m$, $n\in Z$ such that $f(2k)\neq 0, f(2l)\neq 0$,  $ f(2m+1)\neq 0$, $f(2n+1)\neq 0$, where the product  $(k-l)(m-n)klmn\neq 0$. Then we have

\vspace{2mm}1) $f(2k+2l)\neq 0$,  \quad 2) $f(2k+2m)\neq 0$, \quad 3) $f(2k+2m+1)\neq 0$,

\vspace{2mm}4) $f(2m+2n+1)\neq 0$, \quad 5)  $f(1-2k+2m)\neq 0$, $k\neq -m$,

\vspace{2mm}6) $f(4k)\neq 0$, \quad 7) $f(2m+2n+2k+1)\neq 0$,

\vspace{2mm} 8) $f(2m+2k+2l)\neq 0$, \quad
9) $f(2k-2m)\neq 0$, $k\neq -m$,

\vspace{2mm}10) $f(1-2k-2m)\neq 0$,   \quad
11) $f(1-4k)\neq 0$.

\mlabel{cor:cor2}
\end{coro}

\begin{proof}  The results 1), 2), 3),  4), 5) and 6) follow from Corollary \mref{cor:cor1} and $f(0)=-f(1)=1$,  the  results
7) and 8) follow from Lemma  \mref{lem:rbd5}, and  the results  9), 10) and 11) follow from Theorem \mref{thm:thm2}.
\end{proof}

\begin{lemma} Let  $R : A_{\omega} \longrightarrow A_{\omega}$ be a linear map defined by
Eq.( \mref{eq:rf31}). If $f$ satisfies $f(0)=-f(1)=1$ and that there exist finite distinct integers $m_i$ such that
$f(m_{i})\neq 0$ and $f(m)=0$ for $ m\in Z$ and $m\neq m_{i}$, then $R$ is not a homogeneous Rota-Baxter operator on $A_{\omega}$, where $1\leq i\leq t$, and $ m, m_i\neq 0, 1$.
\mlabel{lem:rbd6}
\end{lemma}

\begin{proof} If $R$ is a homogeneous Rota-Baxter operator. Thanks to Theorem \mref{thm:thm2},  $f(1-m_i)=-f(m_i)\neq 0$ for $1\leq i\leq t$. We obtain $t\geq 2$. Without loss of generality,   suppose $m_1$ is odd, then $m_2=1-m_1$ is even and $f(m_2)\neq 0$. Thanks to the result 6) in Corollary \mref{cor:cor2},  $f(2nm_2)\neq 0$ for all $n\in Z$.
It contradicts $t < \infty$. It follows the result.

\end{proof}

\begin{lemma} Let  $R : A_{\omega} \longrightarrow A_{\omega}$ be a linear map defined by
Eq.( \mref{eq:rf31}). If $f$ satisfies $f(0)=-f(1)=1$ and that there exist finite distinct integers $m_i$ such that
$f(m_{i})=0$ and $f(m)\neq 0$ for $ m\in Z$ and $m\neq m_{i}$, then $R$ is not a homogeneous Rota-Baxter operator on $A_{\omega}$, where $1\leq i\leq t$, and $m_i\neq 0, 1$.
\mlabel{lem:rbd7}
\end{lemma}

\begin{proof} If $R$ is a homogeneous Rota-Baxter operator. Thanks to Theorem \mref{thm:thm2},  $f(1-m_i)=-f(m_i)= 0$ for $1\leq i\leq t$, and for all $n\neq m_i$, $f(n)=-f(1-n)\neq 0$.
It follows that there exist infinite odd $2l+1\in Z$ such that $f(2l+1)\neq 0$. Without loss of generality, suppose $m_1$ is even, then by Corollary \mref{cor:cor1},  there exist infinite odd $2l+1\in Z$, such that
$f(2l+2m_1)=0$. It contradicts to $t <\infty$.

\end{proof}

\begin{theorem}  Let  $R : A_{\omega} \longrightarrow A_{\omega}$ be a linear map defined by
Eq.( \mref{eq:rf31}). If $R$ is a homogeneous Rota-Baxter operator on $A_{\omega}$ with $f(0)=-f(1)=1$ and that there exists $m\neq 0, 1$ such that
$f(m)\neq 0$. Then there exists a positive integer $m_0$ such that for  $m\in Z$, $f(m)\neq 0$ if and only if
$ m\in  W=\{2m_0k | k \in Z\}\cup\{1-2m_0k | k \in Z\}.$
\mlabel{thm:thm3}

\end{theorem}

\begin{proof}  From  Theorem \mref{thm:thm2}, there exists $W=\{2x_k | k\in Z\}\cup \{1-2x_k | k\in Z\}\subset Z$ satisfying that $f(m)\neq 0$ if and only if  $m\in W$.
Thanks to Lemma \mref{lem:rbd6} and   \mref{lem:rbd7},  $W$ is an infinite subset of $Z$. From  Corollary \mref{cor:cor2}, we can suppose that for all $k, s\in Z$, $2x_k <2x_s$ if  and only if $ k <s$,  and
 $x_{-1}< 0 <x_{1} $.

By the result 2) of Corollary \mref{cor:cor2} and  $2x_{2}\in W, -2x_{1}+1\in W$,  we have $2(x_{2}-x_{1})\in W$. Thanks to $0<x_{2}-x_1<x_{2}$,  $x_2=2x_1$.

Now suppose that $x_{k}-x_{k-1}=x_{1}$, for $k> 0$, that is, $x_{k}=kx_{1}$. Since $2x_{k+1}\in W, 2x_{k-1}\in W, -2x_{k}+1\in W$, by the result 8) in Corollary \mref{cor:cor2}, we have $2(x_{k+1}-x_{k}+x_{k-1})=2(x_{k+1}-x_{1})\in W$. Thanks to  $x_{k+1}-x_{1}<x_{k+1}$ and $x_{k-1}=x_k-x_1<x_{k+1}-x_{1}$,   $x_{k-1} < x_{k+1}-x_{1} < x_{k+1}$. Therefore,  $x_{k+1}-x_1=x_{k}$, that is,  $x_{k+1}=x_k+x_1=(k+1)x_{1}$.

By the completely similar discussion, we have that  for all $k<0$, $x_k=-kx_{-1}$.

Since $2x_{-1}\in W, 2x_{1}\in W$, from the result 1) in Corollary \mref{cor:cor2}, we have that $2(x_{-1}+x_{1})\in W$. From $x_{-1} <0 <x_{1}$, and
$x_{-1} < x_{-1}+x_{1} < x_1$, we have $x_{-1}+x_{1}=0$, that is,  $x_{-1}=-x_1.$ Denote $m_0=x_1$. Then for all $2x_k\in W$, $2x_k=2m_0k$.

\end{proof}

\vspace{2mm}

For positive integer $m_0$, denote
$$W_{m_0}=\{2m_{0}k~ |~ k\in Z\}\cup \{1-2m_{0}k~ |~ k\in Z\}.$$
 If $f$ satisfies that $f(m)\neq 0$ if and only if $m\in W_{m_0}$, then $W_{m_0}$ is called an {\bf $m_{0}-$ supporter of $R$.}

\begin{coro} Let $R$ be a Homogeneous Rota-Baxter operator with $f(0)=-f(1)=1$. If there exist integer $k$ such that $f(2k)\neq 0$, then we have
$f(-2k)\neq 0$ and  $f(1+2k)\neq 0$.

\mlabel{cor:cor3}
\end{coro}

\begin{proof}
The result follows from Theorem \mref{thm:thm3} and Theorem  \mref{thm:thm2}, directly.
\end{proof}

\begin{lemma}
Let $R$ be a homogeneous Rota-Baxter operator with  $f(0)=-f(1)=1$, and $W_{m_0}$ be its $m_0$-supporter.
 Then   we have $f(2m_0)\neq \frac{1}{2}$, and for all $k, k_1, k_2, k_3\in Z$,

\begin{equation}
\frac{1}{2f(2m_{0}k)}+\frac{1}{2f(-2m_{0}k)}=1, ~ \frac{1}{2f(2m_{0}k)}-\frac{1}{2f(1+2m_{0}k)}=1,
\mlabel{eq:rf36}
\end{equation}

\begin{equation}
\frac{1}{f(2m_{0}k_{2})}+\frac{1}{f(2m_{0}k_{3})}=\frac{1}{f(2m_{0}k_{1})}+\frac{1}{f(-2m_{0}k_{1}+2m_{0}k_{2}+2m_{0}k_{3})}, \quad  k_2\neq k_3.
\mlabel{eq:rf37}
\end{equation}

\mlabel{lem:rbd8}
\end{lemma}

\begin{proof} By the result 4) in Corollary \mref{cor:cor1}, for all $k\in Z$ and $k\neq 0$, we have

$$f(2m_{0}k)f(-2m_{0}k)=f(2m_{0}k)+f(-2m_{0}k)-f(2m_{0}k)f(-2m_{0}k).$$
It follows Eq. (\mref{eq:rf36}), and $f(2m_0)\neq \frac{1}{2}$.

 According to  Lemma \mref{lem:rbd5} and Theorem \mref{thm:thm2}, for all $m, n\in Z$ and $m\neq n$, we have

$$-f(2l)f(2m)f(2n)=\{-f(2l)f(2m)-f(2l)f(2n)+f(2m)f(2n)\}f(-2l+2m+2n).$$
Then in the case $l=m_{0}k_{1}, m=m_{0}k_{2}, n=m_{0}k_{3}$,  we obtain  Eq. (\mref{eq:rf37}).

\end{proof}

\begin{theorem}
Let $R: A_{\omega}\rightarrow A_{\omega}$ be a linear map defined as Eq. (\mref{eq:rf31}) with  $f(0)=-f(1)=1$. Then $R$ is a homogeneous Rota-Baxter operator on $A_{\omega}$
if and only if $f(m)=0$ for all $m\in Z$, $m\neq 0, 1$; or  there exists a positive integer $m_0$ and an element $a\in F$,  $a\neq \frac{k-1}{k}$ for $k\in Z-\{0\}$, such that $W_{m_0}$ is an $m_0$-supporter  of $R$ and

\begin{equation}
f(2m_0k)=-f(1-2m_0k)=\frac{1}{ka-(k-1)}, ~\forall ~  k\in Z.
\mlabel{eq:rf38}
\end{equation}

Further, in the case $m_0=1$, $R$ is an invertible   Rota-Baxter operator on $A_{\omega}$, therefore, $R^{-1}$ is an invertible derivation of
$A_{\omega}$, and
$$R^{-1}(L_{2k})=(ka-(k-1))L_{2k},\quad  R^{-1}(L_{1-2k})=(-ka+(k-1))L_{1-2k},\quad \forall k\in Z.$$

\mlabel{thm:thm4}
\end{theorem}

\begin{proof}
If $R$ is a  homogeneous Rota-Baxter operator on $A_{\omega}$ and there exists $m\neq 0, 1$ such that $f(m)\neq 0$, then by Theorem \mref{thm:thm3},
there exists a positive integer $m_0$ such that $W_{m_0}$ is an $m_0$-supporter  of $R$.  Suppose $f(2m_0)=\frac{1}{a}$, then by Lemma \mref{lem:rbd8}
$a\neq 2$.

Now suppose that for positive integer $k$ satisfies $f(2m_0k)=\frac{1}{ka-(k-1)}$. By  Lemma \mref{lem:rbd8},
$$
\frac{1}{f(2m_0(k+1))}+1=\frac{1}{f(2m_0k)}+\frac{1}{f(2m_0)}=ka-(k-1)+a,
$$
that is,  $f(2m_0(k+1))=\frac{1}{(k+1)a-k}$, and $a\neq \frac{k-1}{k}$.

Since $$
\frac{1}{f(2m_0)}+\frac{1}{f(-2m_0)}=2,
$$
we have $f(-2m_0)=\frac{1}{2-a}=\frac{1}{-a-(-1-1)}.$ Now suppose that for negative integer $k$,   $f(2m_{0}k)=\frac{1}{ka-(k-1)}$. From
$$
\frac{1}{f(2m_0(k-1))}+1=\frac{1}{f(2m_0k)}+\frac{1}{f(-2m_0)}=ka-(k-1)+2-a,
$$
we have
$$f(2m_0(k-1))=\frac{1}{(k-1)a-(k-2)}, \quad \text{ and }\quad a\neq \frac{k-2}{k-1}.$$ It follows Eq. (\mref{eq:rf38}).

Conversely, since for all $2l, 2m, 2n\notin W_{m_{0}}, l\neq m$,

$$f(\pm 2l)=f(\pm 2m)=f(\pm 2n)=0, \quad f(1\pm 2l)=f(1\pm 2m)=f(1\pm 2n)=0, $$
the identity (\mref{eq:rf32}) holds. So we only need to prove that Eq. (\mref{eq:rf32}) holds for the following cases.

\vspace{2mm}1) The case
 $2l, 2m\notin W_{m_0}, 2n \in W_{m_0}$, $l\neq m$.
 By Theorem \mref{thm:thm3} and Theorem \mref{thm:thm2}

 \vspace{2mm}$f(\pm 2l)=f(\pm 2m)=$
$f(1\pm 2l)=f(1\pm 2m)=0,$
 $f(\pm 2m\pm 2n\pm 2l)=f(\pm 2m\pm 2n\pm 2l+1)=0$.
  Then Eq. (\mref{eq:rf32}) holds.

 \vspace{2mm} 2) For the case
 $2l\notin W_{m_0}, 2m, 2n \in W_{m_0}$, and  $m\neq n$.
 We have  $f(\pm 2l)=f(1\pm 2l)=0$,  $f(\pm 2l\pm 2m\pm 2n)=f(1 \pm 2l\pm 2m\pm 2n)=0.$
   Then Eq. (\mref{eq:rf32}) holds.

\vspace{2mm}3) For the case
 $2l, 2m, 2n \in W_{m_0}$, $l\neq m, l\neq n$ and $m\neq n$ .
Suppose
 $2l=2m_0k_1,$ $ 2m=2m_0k_2,$ $ 2n=2m_0k_3\in W_{m_0}$.  From
$$
f(1-2l)f(2m)f(2n)=\frac{-1}{k_1a-(k_1-1)}\frac{1}{k_2a-(k_2-1)} \frac{1}{k_3a-(k_3-1)},
$$

$$
(f(1-2l)f(2m)+f(1-2l)f(2n)
+f(2m)f(2n))f(2(m+n-l))
$$
$$
=(\frac{-1}{k_1a-(k_1-1)} \frac{1}{k_2a-(k_2-1)}+\frac{-1}{k_1a-(k_1-1)}\frac{1}{k_3a-(k_3-1)}
$$
$$
+\frac{1}{k_2a-(k_2-1)} \frac{1}{k_3a-(k_3-1)}) \frac{1}{(-k_1+k_2+k_3)a-(-k_1+k_2+k_3-1)}
$$
$$
=\frac{-1}{k_1a-(k_1-1)} \frac{1}{k_2a-(k_2-1)}\frac{1}{k_3a-(k_3-1)},
$$

$$
f(1-2l)f(1-2m)f(2n)=\frac{-1}{k_1a-(k_1-1)} \frac{-1}{k_2a-(k_2-1)}  \frac{1}{k_3a-(k_3-1)},
$$

$$
(f(1-2l)f(1-2m)+f(1-2l)f(2n)+f(1-2m)f(2n))f(1-2(l+m-n))
$$
$$
=(\frac{-1}{k_1a-(k_1-1)}\frac{-1}{k_2a-(k_2-1)}+\frac{-1}{k_1a-(k_1-1)} \frac{1}{k_3a-(k_3-1)}
$$
$$
+\frac{-1}{k_2a-(k_2-1)}\frac{1}{k_3a-(k_3-1)}) \frac{-1}{(k_1+k_2-k_3)a-(k_1+k_2-k_3-1)}
$$
$$
=\frac{1}{k_1a-(k_1-1)}\frac{1}{k_2a-(k_2-1)}  \frac{1}{k_3a-(k_3-1)},
$$
\\
identity (\mref{eq:rf32}) holds.

Summarizing above discussion, we obtain the result.
\end{proof}

 Let $m_0=1$, and $a=3$. By Theorem \mref{thm:thm4}, the  linear map $R: A_{\omega}\rightarrow A_{\omega}$ defined by
$$R(L_{2k})=\frac{1}{3k-(k-1)}L_{2k}=\frac{1}{2k+1}L_{2k}, \quad R(L_{1-2k})=-\frac{1}{2k+1}L_{1-2k},  k\in Z,$$
 is a homogeneous Rota-Baxter operator on $A_{\omega}$, and  $R$ is a invertible Rota-Baxter operator. Therefore,
 $D=R^{-1}:  A_{\omega}\rightarrow A_{\omega}$ satisfying $$D(L_{2k})=(2k+1)L_{2k}, \quad D(L_{1-2k})=-(2k+1)L_{1-2k},  k\in Z,$$
 is an invertible  derivation of $A_{\omega}$.

If $m_0=3$, and $a=\sqrt{2}$. Then  the linear map $R: A_{\omega}\rightarrow A_{\omega}$ defined by
$$R(L_{6k})=\frac{1}{\sqrt{2}k-(k-1)}L_{6k}, \quad R(L_{1-6k})=-\frac{1}{\sqrt{2}k-(k-1)}L_{1-6k},  k\in Z,$$
 and others are zero, is a homogeneous Rota-Baxter operator on $A_{\omega}$.  But $R$ is degenerate.

\vspace{3mm}\noindent 3.2.2  {\bf Homogeneous Rota-Baxter operators with  $f(0)=f(1)=0 $}

In this section we discuss the case $f(0)=f(1)=0$.

\begin{lemma}  Let $R: A_{\omega} \rightarrow A_{\omega}$ be a homogeneous Rota-Baxter operator  on $A_{\omega}$ with
$f(0)=f(1)=0$. Then $R$  satisfies that for all $l, m, n\in Z$,

\vspace{2mm}1) $f(2l+1)f(2m+1)f(2l+2m+1)=0,\quad  l\neq m$.

\vspace{2mm}2) $f(2m+1)f(2n)f(2m+2n+1)=0, \quad ~ m\neq 0$.

\vspace{2mm}3) $f(2l+1)f(2m)f(2l+2m)=0,~  \quad~   m\neq 0$.

\vspace{2mm}4) $f(2m)f(2n)f(2m+2n)=0, ~  \quad ~ m\neq n$.
\mlabel{lem:rbd01}
\end{lemma}

\begin{proof} The result follows from  $D(2l+1, 2m+1, 0)\neq 0$, $D(1, 2m+1, 2n)\neq 0$, $D(2l+1, 2m, 0)\neq 0$, $D(1, 2m, 2n)\neq 0$,
and Lemma \mref{lem:rbd5}.

\end{proof}

\begin{coro} Let $R: A_{\omega} \rightarrow A_{\omega}$ be a  homogeneous Rota-Baxter operator  with
$f(0)=f(1)=0$, and there exist $k, l, m, n\in Z$ such that   $(k-l)(m-n)klmn\neq 0$, $f(2k)\neq 0,$ $ f(2l)\neq 0, $ $f(2m+1)\neq 0, $ $f(2n+1
)\neq 0. $ Then  we have

\vspace{2mm} 1) $f(2k+2l)=0$,\quad  2) $f(2k+2m)=0$, \quad 3) $f(2k+2m+1)=0$,

\vspace{2mm} 4) $f(2m+2n+1)=0$, \quad 5) $f(2m+2n+2k+1)\neq 0$,

\vspace{2mm} 6) $f(2m+2k+2l)\neq 0$, \quad 7) $f(2k-2m)=0$, $k\neq -m$, \quad 8) $f(4k)=0$.
\mlabel{cor:cor01}
\end{coro}

\begin{proof} The result 1), 2), 3) and 4)  follow from the result 4), 3), 2) and 1) in Lemma \mref{lem:rbd01},  respectively.
 The result 5) and 6) follow from   Eq. (\mref{eq:rf33}) and Eq. (\mref{eq:rf34}),  respectively. The result 7) and 8) follow from the result 4) and 3) in Lemma \mref{lem:rbd01}, respectively.

\end{proof}

\begin{theorem}

Let $R: A_{\omega} \rightarrow A_{\omega}$ be  a homogeneous Rota-Baxter operator  with
$f(0)=f(1)=0$, and there exist $m_1, \cdots, m_s\in Z$ such that $f(m_i)\neq 0$ and $f(m)=0$ for all $m\neq m_i$, where  $m_i\neq 0, 1$,  $1\leq i\leq s$. Then we
have

\vspace{2mm}1) $s=1$, and then we can suppose $f(m_1)=1$, $f(m)=0$ for all $m\in Z$, $m\neq m_1$.

\vspace{2mm} 2) $s=2$ and $m_1+m_2=1$, so we can suppose that $f(m_1)=1$, $f(1-m_1)=b$, and $f(m)=0$ for all $m\in Z$, $m\neq m_1, 1-m_1$, where  $b\in F$, $b\neq 0$.
\mlabel{thm:thm01}
\end{theorem}

\begin{proof} First, if there exists only one $m_1\in Z$, $m_1\neq 0, 1$ such that $f(m_1)\neq 0$ and $f(m)=0$ for all $m\in Z$ and $m\neq m_1$. By Lemma \mref{lem:rbd5} and a direct  computation,  $R$
is a homogeneous Rota-Baxter operator. Thanks to Lemma  \mref{lem:rbd2}, we can suppose $f(m_1)=1.$

Second, if there exist only two distinct integers $m_1, m_2$ satisfying $m_1+m_2=1$ and  $m_1\neq 0, 1$ such that $ f(m_1)\neq 0$, $f(m_2)\neq 0$ and $f(m)=0$ for all $m\in Z$
and $m\neq m_1$, $m\neq m_2$. Then for all $m\in Z$,  we have $D(m_1, m_2, m)\neq 0.$ By a direct computation, for all $l, m, n\in Z$, we have that  $f(l)$, $f(m) $ and $f(n)$
satisfy Eq. (\mref{eq:rf33}) and Eq. (\mref{eq:rf34}). Therefore, $R$ is a homogeneous Rota-Baxter operator.
 By Lemma  \mref{lem:rbd2}, we can suppose $f(m_1)=1$, $f(m_2)=f(1-m_1)=b$, where $b\in F$ and $b\neq 0$.

Third, if $R$ is a homogeneous Rota-Baxter operator satisfying that there exist two distinct integers $m_1, m_2$ such that $f(m_i)\neq 0$ and $f(m)=0$, for all $m\in Z$ and $m\neq m_i,$ $i=1, 2$, where   $m_1, m_2\neq 0, 1$.
 Then there exists $m\in Z$ such that $D(m_1, m_2, m)\neq 0$. Thanks to Lemma \mref{lem:rbd5},
$$f(m_{1}+m_{2}+m-1)=0. $$
Then $m_{1}+m_{2}+m-1\neq m_{1}$ and $m_{1}+m_{2}+m-1\neq m_{2}$, that is, $m\neq 1-m_1$ and $m\neq 1-m_2$. It follows that  $1-m_{1}=m_{2}$.

Lastly,  if $R$ is a homogeneous Rota-Baxter operator satisfying $f(m_i)\neq 0$,  and $f(m)=0$ for all $m\neq m_i$, $1\leq i\leq s$, $s\geq 3$. Then
 for every  $1\leq i\leq s$, $f(1-m_i)\neq 0$.

In fact,  if $f(1-m_1)=0.$ From  $D(m_1, m_2, 1-m_1)\neq 0$, Eq. (\mref{eq:rf33}) and  Eq. (\mref{eq:rf34}), we have  $f(m_{1}+m_{2}+(1-m_1)-1)=f(m_2)=0$. Contradiction. Therefore,
$f(1-m_1)\neq 0.$
From $s\geq 3$, and  similar discussions,  we have that  $f(1-m_i)\neq 0$, for  $1\leq i\leq s$,.

Therefore, $s$ is even and $s\geq 4$ and we can suppose $m_1 < \cdots < m_i < m_{i+1} < \cdots < m_{s}. $ Then there exists $m\in Z$, $m\neq 0, 1$ such that $f(m)=0$ and $D(m_1, m_2, m)\neq 0$.
Thanks to  Eq. (\mref{eq:rf33}) and  Eq. (\mref{eq:rf34}), $f(m_1+m_2+m-1)=0$. Then $m_1+m_2+m-1\neq m_s$, that is,  $m\neq m_s-m_1-m_2+1$. By the above discussion and $s\geq 4$,
there exists $i\geq 3$ such that $m_s-m_1-m_2+1=1- m_i$. We obtain that $m_1+m_2=m_i+m_s$.   Contradiction.

Summarizing above discussion, we obtain the result.
\end{proof}

\begin{lemma}
Let $R: A_{\omega} \rightarrow A_{\omega}$ be  a homogeneous Rota-Baxter operator  with
$f(0)=f(1)=0$, and satisfy that there exist infinite $m\in Z$ such that $f(m)\neq 0$. Then there exist infinite $n\in Z$ such that $f(n)=0$,
and for all $m\in Z$,   if $f(m)\neq 0$, then $f(1-m)\neq 0$, and $$f(m)+f(1-m)=0.$$
\mlabel{lem:rbd02}
\end{lemma}

\begin{proof} If there exists $m\in Z$ such that $f(m)\neq 0$, but $f(1-m)=0$. Then for all $n\in Z$ and  $n\neq m, 1-m$,  by   Eq. (\mref{eq:rf33}), Eq. (\mref{eq:rf34}) and $D(m, n, 1-m)\neq 0$, we have $f(m+n+1-m-1)=f(n)=0$. Contradiction. Therefore, if $f(m)\neq 0$, then $f(1-m)\neq 0$. Thanks to the result 8) in  Corollary \mref{cor:cor01}, there exist
infinite $n\in Z$ such that $f(n)=0$.

Now for distinct $2m, 2n\in Z$, $f(2m)\neq 0$ and $f(2n)\neq 0$ and $m\neq n$,  by  Eq. (\mref{eq:rf34}),

  \vspace{2mm}$f(1-2m)f(2m)f(2n)$

 \vspace{2mm} $=((f(1-2m)f(2m)+f(1-2m)f(2n)+f(2m)f(2n))f(1-2m+2n+2m)$.

It follows  $f(1-2m)+f(2m)=0$ for all $m\in Z$. The proof is completed.

\end{proof}

\begin{theorem}   Let $R: A_{\omega} \rightarrow A_{\omega}$ be  a homogeneous Rota-Baxter operator with
$f(0)=f(1)=0$, and there exist infinite $m\in Z$ such that $f(m)\neq 0$. Then there exist positive integer
$m_0$ and $s_0$ satisfying $1\leq s_0 < m_0$ such that  $f(m)\neq 0$ if and only if $m\in W=\{2m_{0}k+2s_{0}| ~~ k\in Z\}\cup \{1-2m_{0}k-2s_{0}| ~~ k\in Z\}$.
\mlabel{thm:thm02}
\end{theorem}

\begin{proof}  By Lemma \mref{lem:rbd02}, we can suppose that
$W=\{2x_k | k\in Z\}\cup \{1-2x_k | k\in Z\}$ is set of integers satisfying that $f(m)\neq 0$ if and only if $m\in W$. By Lemma \mref{lem:rbd02} and Corollary \mref{cor:cor1}, we can suppose that
for $2x_k, 2x_s\in W$,  $2x_k < 2x_s$ if and only if $k <s$ and  $2x_{-1}<0$, $2x_0 >0$.

Denote $x_{1}-x_{0}=m_{0}, x_{2}-x_{1}=m_{1}$. Then  $m_{0}>0, m_{1}>0$. From  $2x_{0}\in W,$ $ -2x_{1}+1\in W,$ $ 2x_{2}\in W$ and the result  6)  in  Corollary \mref{cor:cor01}, we have
$$2(x_{2}-x_{1}+x_{0})=2(m_{0}+x_{0}-x_{0}+x_{0}-m_{1})=2(x_{2}-m_{0})\in W.$$

Since $x_0=x_{1}- m_{0}<x_{2}-m_{0} <x_{2}$,   $x_{2}-m_{0}=x_{1}$, that is, $m_{1}=m_{0}$.

Now suppose $x_{k}-x_{k-1}=m_{0}$ for $k>0$. Denote  $x_{k+1}-x_{k}=m_{k}$. According to the result 6) of Corollary \mref{cor:cor01},  we have $$2(x_{k+1}-x_{k}+x_{k-1})=2(m_{k}+x_{k}-x_{k}+x_{k}-m_{0})=2(x_{k+1}-m_{0})\in W.$$
 Thanks to $x_{k-1}=x_k-m_{0}<x_{k+1}-m_0 <x_{k+1}$, $x_{k+1}-m_0=x_k$, that is, $m_{k}=m_{0}$.
  Therefore, $2x_k=2km_0+2x_0, ~ k > 0,   k\in Z.$

 Similar discussion we have $2x_k=2km_0+2x_0, $ for all $ k <0, k\in Z.$

Therefore, $W=\{ 2km_0+2x_0 | ~ k\in Z, x_0 > 0\},$  where ~ $m_0 > 0.$

By Lemma \mref{lem:rbd02} and  the result  1) in Corollary \mref{cor:cor01}, $2x_1+2x_{-1}=2x_0+2x_0\not\in W$, that is, $m_0$ is not a factor of $x_0.$
So there exist  integers $s_0$ and $q$ such that $1\leq s_0 < m_0$ and $x_0=qm_0+s_0$.

Therefore,   $2x_k=2(k+q)m_0+2s_0$, for all $k\in Z$.
It follows the result.
\end{proof}

For positive integer $m_0$ and $s_0$ with  $1\leq s_0 < m_0$, denote
$$W_{m_0, s_0}=\{2m_{0}k+2s_{0}~ |~ k\in Z\}\cup \{1-2m_{0}k-2s_{0}~ |~ k\in Z\}.$$
 If $f$ satisfies that $f(m)\neq 0$ if and only if $m\in W_{m_0, s_0}$, then $W_{m_0, s_0}$ is called an {\bf $(m_{0}, s_{0})-$ supporter of $R$.}
  By Lemma \mref{lem:rbd2}, {\bf we can suppose that  $f(2s_0)=1.$ }

\begin{lemma} Let $R: A_{\omega} \rightarrow A_{\omega}$ be  a Homogeneous Rota-Baxter operator with $(m_{0}, s_{0})-$ supporter $W_{m_0, s_0}$, and
$f(0)=f(1)=0$. Then for all $k_i\in Z$, and $k_i\neq k_j$,  for $1\leq i\neq j\leq 3$, we have
 \begin{equation}
\frac{1}{f(2m_{0}k_{1}+2s_{0})}+\frac{1}{f(2m_{0}k_{2}+2s_{0})}+\frac{1}{f(2m_{0}k_{3}+2s_{0})}
=\frac{1}{f(2m_{0}(k_{1}+k_{2}-k_{3})+2s_{0})}
\mlabel{eq:rf01}
\end{equation}
$$
+\frac{1}{f(2m_{0}(k_{1}-k_{2}+k_{3})+2s_{0})}+\frac{1}{f(2m_{0}(-k_{1}+k_{2}+k_{3})+2s_{0})}.
$$

Therefore,
 \begin{equation}
\frac{1}{f(2m_{0}k+2s_{0})}+\frac{1}{f(-2m_{0}k+2s_{0})}=2,
\mlabel{eq:rf02}
\end{equation}
and $f(2m_{0}k+2s_{0})\neq \frac{1}{2}$~  for all ~ $ k\in Z$.
\mlabel{lem:rbd03}
\end{lemma}

\noindent{\bf Proof} By  Lemma \mref{lem:rbd5} and Lemma \mref{lem:rbd02}, for all $k_1, k_2, k_3\in Z$ and $ k_{1}\neq k_{2}$, we have

\vspace{2mm}$
f(2m_{0}k_{1}+2s_{0})f(2m_{0}k_{2}+2s_{0})f(2m_{0}k_{3}+2s_{0})$

\vspace{2mm}$=(-f(2m_{0}k_{1}+2s_{0})f(2m_{0}k_{2}+2s_{0})+f(2m_{0}k_{1}
+2s_{0})f(2m_{0}k_{3}+2s_{0})$

\vspace{2mm}$+f(2m_{0}k_{2}+2s_{0})f(2m_{0}k_{3}+2s_{0}))f(2m_{0}(k_{1}+k_{2}-k_{3})+2s_{0})\neq 0.
$

Therefore,
$$ \frac{1}{f(2k_1m_0+2s_0)}+\frac{1}{f(2k_2m_0+2s_0)}-\frac{1}{f(2k_3m_0+2s_0)}=\frac{1}{f(2(k_1+k_2-k_3)m_0+2s_0)}.$$
For the case $k_1=-k_2$ and $k_3=0$, we obtain Eq.(\mref{eq:rf02}).

Similarly, for $k_1\neq k_3$, we have
$$ \frac{1}{f(2k_1m_0+2s_0)}+\frac{1}{f(2k_3m_0+2s_0)}-\frac{1}{f(2k_2m_0+2s_0)}=\frac{1}{f(2(k_1+k_3-k_2)m_0+2s_0)},$$
and for $k_2\neq k_3$, we have
$$ \frac{1}{f(2k_2m_0+2s_0)}+\frac{1}{f(2k_3m_0+2s_0)}-\frac{1}{f(2k_1m_0+2s_0)}=\frac{1}{f(2(k_2+k_3-k_1)m_0+2s_0)},$$
 It follows Eq. (\mref{eq:rf01}).
The result follows.

\begin{theorem}
Let $R: A_{\omega}\rightarrow A_{\omega}$ be a linear map defined as Eq. (\mref{eq:rf31}) which satisfies that there exist infinite $m\in Z$ such that $f(m)=f(0)=f(1)=0$. Then $R$ is a homogeneous Rota-Baxter operator on $A_{\omega}$
if and only if there exist positive integer   $m_0$ and $s_0$, and $a\in F$,  such that $W_{m_0, s_0}$ is an $(m_0, s_0)$-supporter of $R$, and

\begin{equation}
f(2m_0k+2s_0)=-f(1-2m_0k-2s_0)=\frac{1}{ka-(k-1)}, ~\forall ~  k\in Z,
\mlabel{eq:rf03}
\end{equation}
 where $1\leq s_0 < m_0$ and $a\neq \frac{k-1}{k}$, for all $k\in Z$ and $k\neq 0$.
\mlabel{thm:thm03}
\end{theorem}

\begin{proof} The  proof is completely similar to Theorem \mref{thm:thm4}.

\end{proof}

   Let $m_0=7$, $a=2$, and $s_0=2$. By Theorem \mref{thm:thm03}, the linear map $R: A_{\omega}\rightarrow A_{\omega}$ defined by for all $k\in Z$,
$$R(L_{14k+4})=\frac{1}{2k-(k-1)}L_{14k+4}=\frac{1}{k+1}L_{14k+4}, ~~ \quad R(L_{-14k-3})=-\frac{1}{k+1}L_{-14k-3}, $$
and others are zero, is a Homogeneous Rota-Baxter operator of weight $0$ with $(7, 2)$-supporter $$W_{7, 2}=\{ 14k+4 |~ k\in Z\}\cup \{ -14k-3 |~ k\in Z\}.$$

If $m_0=4,$ $s_0=3$ and $a=\frac{3}{5}$, then the linear map $R: A_{\omega}\rightarrow A_{\omega}$ defined by for all $k\in Z$,
$$R(L_{8k+6})=\frac{5}{5-2k}L_{8k+6}, ~~ \quad R(L_{-8k-5})=-\frac{5}{5-2k}L_{-8k-5}, $$
and others are zero, is a homogeneous Rota-Baxter operator of weight $0$ with the $(4, 3)$-supporter $$W_{4, 3}=\{ 8k+6 |~ k\in Z\}\cup \{-8k-5 |~ k\in Z\}.$$

\subsection{$3$-Lie algebras constructed by $A_{\omega}$ and homogeneous Rota-Baxter operators}
\mlabel{ss:3lialg}

In the  study of $3$-Lie algebras, we know that construction of $3$-Lie algebras from known algebras  is always interesting.
So in this section, we construct $3$-Lie algebras from the $3$-Lie algebra $A_{\omega}$ and the homogeneous Rota-Baxter operators.

Let $(A, [,,]) $ be a $3$-Lie algebra and $R$ be a Rota-Baxter with  weight $\lambda$. Using the notation in  Eq. (\mref{eq:rbn}), we define a ternary operation $[ , , ]_R$ on $A$ by
\begin{equation}
{[}x_{1}, x_{2}, x_{3}]_{R}= \sum_{\emptyset \neq I \subseteq [3]} \lambda^{|I|-1}[\widehat{R}_{I}(R(x_{1})), \widehat{R}_{I}(R(x_{2})),\widehat{R}_{I}(R(x_{3}))],~~ \forall x, y, z \in A.
\mlabel{eq:rf04}
\end{equation}

Therefore, in the case $\lambda=0$, we have

\begin{equation}
{[}x_{1}, x_{2}, x_{3}]_{R} =[R(x), R(y), z]+[R(x), y, R(z)]+[x, R(y), R(z)], ~~ \forall x, y, z \in A.
\mlabel{eq:rf05}
\end{equation}

\begin {theorem}\cite{BGL11} Let $(A, [,,]) $ be a $3$-Lie algebra and $R$ be a Rota-Baxter of  weight $\lambda$. Then $(A, [ , , ]_R)$ is $3$-Lie algebra in the multiplication
defined as Eq.  (\mref{eq:rf04}), and $R$ is also a Rota-Baxter operator of it.
\mlabel{thm:thm04}
\end{theorem}

So if $R$ is a homogeneous Rota-Baxter operator of the $3$-Lie algebra $A_{\omega}$ of  weight $0$,  then $(A, [ , , ]_R)$ is  a  $3$-Lie algebra in
the multiplication defined as Eq.  (\mref{eq:rf05}), where $A=A_{\omega}$ as vector spaces, and $R$ is also   a homogeneous Rota-Baxter operator of $(A, [ , , ]_R)$.

\begin{theorem} Let $R: A_{\omega}\rightarrow A_{\omega}$ be a linear map defined as Eq. (\mref{eq:rf31}), then $R$ is a
homogeneous Rota-Baxter operator of weight $0$ on $3$-Lie algebra $A_{\omega}$ if and only if $R$ is the one of the following

\vspace{2mm}$R_{0_1}(L_0)=L_0$, $R_{0_1}(L_1)=bL_1$,   and $R_{0_1}(L_m)=0$,
 for all $m\in Z$, $m\neq 0, 1.$

\vspace{2mm}$R_{0_{2}}(L_{m})=\left\{
  \begin{array}{ll}
  L_{0}, \quad m=0, &\\
  -L_{1}, \quad m=1, &\\
  \frac{1}{ka-(k-1)}L_{2m_{0}k}, m=2m_{0}k\in W_{m_{0}}, &\\
  -\frac{1}{ka-(k-1)}L_{1-2m_{0}k}, m=1-2m_{0}k\in W_{m_{0}},&\\
  0, \quad \text{others}.&
  \end{array}
\right.$

\vspace{2mm}$R_{0_{3}}(L_{m})=\left\{
  \begin{array}{ll}
   \frac{1}{ka-(k-1)}L_{2m_{0}k+2s_{0}}, \quad  m=2m_{0}k+2s_{0}\in W_{m'_{0}, s_{0}},&\\
  -\frac{1}{ka-(k-1)}L_{1-2m_{0}k-2s_{0}},  m=1-2m_{0}k-2s_{0}\in W_{m'_{0}, s_{0}}, &\\
  0, \quad \quad \quad \quad\quad\quad \quad \forall m\notin W_{m'_{0}, s_{0}}.&
    \end{array}
\right.$~

\vspace{2mm}$R_{0_{4}}(L_{m})$=$\left\{
  \begin{array}{ll}
  L_{m_{1}}, ~~m=m_{1},&\\
  0,~~\quad  m\neq m_{1}.&
  \end{array}
\right.$~~~~

\vspace{2mm} $R_{0_{5}}(L_{m})$=$\left\{
  \begin{array}{ll}
  L_{m_{1}}, ~~\quad m=m_{1}, &\\
  bL_{1-m_{1}}, ~~m=1-m_{1}, &\\
  0,~~ \quad \quad m\neq m_{1}, 1-m_{1}.&
  \end{array}
\right.$~~~~
\\where $m_1, m_0, m'_0, s_0\in Z$, $m_1\neq 0, 1$; $m_0>0$;  $1\leq s_0 <m'_0$; $a, b, c\in F$, $a\neq \frac{k-1}{k}$, $b\neq 0$,
$W_{m_0}=\{2m_{0}k~ |~ k\in Z\}\cup \{1-2m_{0}k~ |~ k\in Z\}$, $W_{m'_0, s_0}=\{2m'_{0}k+2s_{0}~ |~ k\in Z\}\cup \{1-2m'_{0}k-2s_{0}~ |~ k\in Z\}.$

 \mlabel{thm:thm04}
\end{theorem}

\begin{proof}  The result follows from  Theorem \mref{thm:thm1}, Theorem \mref{thm:thm4}, Theorem \mref{thm:thm01} and Theorem \mref{thm:thm03}.
\end{proof}

For convenience, denote $\lambda_k=ka-(k-1),$ for all $k\in Z$, $k\neq \frac{k-1}{k}$, and the multiplication $[ , , ]_{R_{0_i}}$ defined as Eq. (\mref{eq:rf05})  by $[ , , ]_i$,  $1\leq i\leq 5$.
Then we obtain $3$-Lie algebras $(A, [ , , ]_{i})$ with the homogeneous Rota-Baxter operators $R_{0_i}$ for $1\leq i\leq 5$, where $A=A_{\omega}$ as vector spaces.
And  we omit
the zero product of basis vectors in the multiplication of $3$-Lie algebras $[A, [  , , ]_i])$, for $1\leq i\leq 5$.

\vspace{2mm}1) $([A, [ , , ]_1)$ with the multiplication

\vspace{2mm}$[L_0, L_1, L_m]_1=c(2m-1+(-1)^m)L_m$, for all $m\in Z, m\neq 0, 1$, $b\in F$, $b\neq 0$.

\vspace{2mm}2) $([A, [ , , ]_2)$ with the multiplication

\vspace{2mm}$ {[}L_{0},L_{1},L_{2m}{]}_2=-4mL_{2m},$

 \vspace{2mm} ${[}L_{0},L_{1},L_{2m+1}{]}_2=-4mL_{2m+1}, $

 \vspace{2mm}  ${[}L_{0},L_{1-2m_{0}k_{1}},L_{2m}{]}_2=\frac{-4m}{\lambda_{k_1}}L_{2m-2m_{0}k_{1}}, $

  \vspace{2mm}${[}L_{0},L_{2m_{0}k_{1}},L_{2m+1}{]}_2=-\frac{4m_0k_1}{\lambda_{k_1}}L_{2m+2m_{0}k_{1}}, $

    \vspace{2mm} $  {[}L_{1},L_{2m_{0}k_{1}},L_{2m}{]}_2=-\frac{4m_0k_1-4m}{\lambda_{k_1}}L_{2m+2m_{0}k_{1}}, $

  \vspace{2mm}  $  {[}L_{1},L_{2m_{0}k_{1}},L_{2m+1}{]}_2=\frac{4m}{\lambda_{k_1}}L_{2m+2m_{0}k_{1}+1}, $

  \vspace{2mm} $  {[}L_{1},L_{1-2m_{0}k_{1}},L_{2m}{]}_2=-\frac{4m_0k_1}{\lambda_{k_1}}L_{2m-2m_{0}k_{1}+1},$

  \vspace{2mm} $ {[}L_{0},L_{2m_{0}k_{1}},L_{1-2m_{0}k_{2}}{]}_2=\frac{-4m_0k_1(\lambda_{k_2}-\lambda_{k_1}-1)}{\lambda_{k_1}\lambda_{k_2}}L_{2m_{0}
  (k_{1}-k_{2})},$

\vspace{2mm}  $
  {[}L_{0},L_{1-2m_{0}k_{1}},L_{1-2m_{0}k_{2}}{]}_2=\frac{4m_0(k_1-k_2)(-\lambda_{k_2}-\lambda_{k_1}+1)}{\lambda_{k_1}\lambda_{k_2}}
  L_{-2m_{0}(k_{1}+k_{2})+1}, $

  \vspace{2mm}$  {[}L_{0},L_{1-2m_{0}k_{1}},L_{2m+1}{]}_2=-\frac{4m+4m_0k_1}{\lambda_{k_1}}L_{2m-2m_{0}k_{1}+1}, $

\vspace{2mm}  $
  {[}L_{1},L_{2m_{0}k_{1}},L_{1-2m_{0}k_{2}}{]}_2=\frac{4m_0k_2(\lambda_{k_1}-\lambda_{k_2}-1)}{\lambda_{k_1}\lambda_{k_2}}
  L_{2m_{0}(k_{1}-k_{2})+1}, $

 \vspace{2mm}$ {[}L_{1},L_{2m_{0}k_{1}},L_{2m_{0}k_{2}}{]}_2=\frac{4m_0(k_1-k_2)(-\lambda_{k_2}-\lambda_{k_1}+1)}{\lambda_{k_1}\lambda_{k_2}}
  L_{2m_{0}(k_{1}+k_{2})}, $

 \vspace{2mm} $
  {[}L_{2m_{0}k_{1}},L_{1-2m_{0}k_{2}},L_{2m}{]}_2=-\frac{4m-4m_0k_1}{\lambda_{k_1}\lambda_{k_2}}L_{2m+2m_{0}(k_{1}-k_{2})},$

 \vspace{2mm} $
  {[}L_{2m_{0}k_{1}},L_{1-2m_{0}k_{2}},L_{2m+1}{]}_2=-\frac{4m+4m_0k_2}{\lambda_{k_1}\lambda_{k_2}}L_{2m+2m_{0}(k_{1}-k_{2})+1}, $

\vspace{2mm}  $
  {[}L_{2m_{0}k_{1}},L_{2m_{0}k_{2}},L_{1-2m_{0}k_{3}}{]}_2=\frac{4m_0(k_1-k_2)(\lambda_{k_3}-\lambda_{k_2}-\lambda_{k_1})}
  {\lambda_{k_1}\lambda_{k_2}\lambda_{k_3}}
  L_{2m_{0}(k_{1}+k_{2}-k_{3})}, $

\vspace{2mm}  $
  {[}L_{2m_{0}k_{1}},L_{1-2m_{0}k_{2}},L_{1-2m_{0}k_{3}}{]}_2=\frac{4m_0(k_2-k_3)(\lambda_{k_1}-\lambda_{k_2}-\lambda_{k_3})}
  {\lambda_{k_1}\lambda_{k_2}\lambda_{k_3}}L_{2m_{0}(k_{1}-k_{2}-k_{3})+1}, $

\vspace{2mm} $ {[}L_{2m_{0}k_{1}},L_{2m_{0}k_{2}},L_{2m+1}{]}_2=\frac{4m_0(k_1-k_2)}
  {\lambda_{k_1}\lambda_{k_2}}L_{2m+2m_{0}(k_{1}+k_{2})}, $

\vspace{2mm} $
  {[}L_{1-2m_{0}k_{1}},L_{1-2m_{0}k_{2}},L_{2m}{]}_2=\frac{4m_0(k_1-k_2)}
  {\lambda_{k_1}\lambda_{k_2}}L_{2m-2m_{0}(k_{1}+k_{2})+1}, $
\\
for all $~2m+1, 2m\in Z$ and $2m, 2m+1 \notin W_{m_{0}}$,  where $m_0\in Z$, $m_0 > 0$,

\vspace{2mm} 2) $([A, [ , , ]_3)$ with the multiplication

\vspace{2mm}$[L_{2m_{0}k_{1}+2s_{0}},L_{2m_{0}k_{2}+2s_{0}},L_{1-2m_{0}k_{3}-2s_{0}}{]}_3
    =\frac{4m_0(k_1-k_2)(\lambda_{k_{3}}-\lambda_{k_{2}}-\lambda_{k_{1}})}{\lambda_{k_1}\lambda_{k_2}\lambda_{k_3}}L_{2m_{0}(k_{1}+k_{2}-k_{3})+2s_{0}},$

  \vspace{2mm}$[L_{2m_{0}k_{1}+2s_{0}},L_{1-2m_{0}k_{2}-2s_{0}},L_{1-2m_{0}k_{3}-2s_{0}}{]}_3=
  \frac{4m_0(k_2-k_3)(\lambda_{k_{1}}-\lambda_{k_{2}}-\lambda_{k_{3}})}{\lambda_{k_1}\lambda_{k_2}\lambda_{k_3}}L_{2m_{0}(k_{1}-k_{2}-k_{3})-2s_{0}+1},$

\vspace{2mm}$[L_{2m_{0}k_{1}+2s_{0}},L_{1-2m_{0}k_{2}-2s_{0}},L_{2m+1}{]}_3  =-\frac{4(m+m_0k_2+s_0)}{\lambda_{k_1}\lambda_{k_2}}L_{2m+2m_{0}(k_{1}-k_{2})+1},$

\vspace{2mm}$
  [L_{2m_{0}k_{1}+2s_{0}},L_{1-2m_{0}k_{2}-2s_{0}},L_{2m}{]}_3
  =\frac{4(m-m_0k_1-s_0)}{\lambda_{k_1}\lambda_{k_2}}L_{2m+2m_{0}(k_{1}-k_{2})},$

\vspace{2mm}$
  [L_{2m_{0}k_{1}+2s_{0}},L_{2m_{0}k_{2}+2s_{0}},L_{2m+1}{]}_3
  =\frac{4m_0(k_1-k_2)}{\lambda_{k_1}\lambda_{k_2}}L_{2m+2m_{0}(k_{1}+k_{2})+4s_{0}},$

\vspace{2mm}$  [L_{1-2m_{0}k_{1}-2s_{0}},L_{1-2m_{0}k_{2}-2s_{0}},L_{2m}{]}_3
  =\frac{4m_0(k_1-k_2)}{\lambda_{k_1}\lambda_{k_2}}L_{2m-2m_{0}(k_{1}+k_{2})-4s_{0}+1},$
  \\ for all $~2m+1, 2m\in Z$ and $2m, 2m+1 \notin W_{m_{0}, s_0}$, where $m_0, s_0\in Z$, $1\leq s_0 < m_0$.

\vspace{2mm} 4) $([A, [ , , ]_4)$  ia an abelian algebras.

\vspace{2mm} 6) $([A, [ , , ]_5)$ with the multiplication

\vspace{2mm}$[L_{m_1}, L_{1-m_1}, L_m]_5=bD(m_1, 1-m_1, m)L_m$, for all $m\in Z$, $m\neq m_1$,  where  $m_1\in Z$, $m_1\neq 0, 1, $  $b\in F$,  $b\neq 0$.

\noindent
{\bf Acknowledgements. } The first author (R.-P. Bai) was supported in part by the Natural
Science Foundation (11371245) and the Natural
Science Foundation of Hebei Province (A2014201006).

\bibliography{}

\end{document}